%% file: main.tex
\documentclass[a4paper,UKenglish]{lipics}

\usepackage{microtype}

\usepackage{pdfsync}
\usepackage[utf8]{inputenc}
\usepackage{hyperref}
\usepackage{enumerate}
\usepackage{amsmath,amscd,amssymb,amsfonts,latexsym,xspace}

\begin{document}

\pagestyle{plain}

\title{Reachability analysis of first-order definable pushdown systems}

\author[1]{Lorenzo Clemente}
\author[1]{Sławomir Lasota}
\affil[1]{University of Warsaw%
}
\authorrunning{L.\,Clemente and S.\,Lasota} 

\Copyright{Lorenzo Clemente and Sławomir Lasota}



\serieslogo{}
\EventShortName{}
\DOI{10.4230/LIPIcs.xxx.yyy.p}

\maketitle

\begin{abstract}
	We study pushdown systems
	where control states, stack alphabet, and transition relation, instead of being finite, are first-order definable in a fixed countably-infinite structure.
	We show that the reachability analysis can be addressed with the well-known saturation technique
	for the wide class of \emph{oligomorphic structures}.
	Moreover, for the more restrictive \emph{homogeneous structures},
	we are able to give concrete complexity upper bounds.
	We show ample applicability of our technique by presenting several concrete examples of homogeneous structures,
	subsuming, with optimal complexity, known results from the literature.
	We show that infinitely many such examples of homogeneous structures can be obtained with the classical \emph{wreath product} construction.
\end{abstract}

\input{defs}

\input{intro}

\input{prelim}

\input{fodef-aut}

\input{saturation}

\input{homogeneous}

\input{examples}
\input{conclusions}

\bibliographystyle{abbrv}
\bibliography{bib}


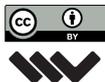
\end{document}

%% file: defs.tex

 \newcommand{\lorenzo}[1]{\todo[color=blue!20]{#1}}

\newtheorem{proposition}{Proposition}
\newtheorem{claim}{Claim}

\newcommand{\mysubsection}[1]{\subparagraph*{#1.}}
\newcommand{\atoms}{\mathbb A}
\newcommand{\vatoms}{\mathbb V}
\newcommand{\btoms}{\mathbb B}
\newcommand{\nat}{\mathbb N}
\newcommand{\N}{\mathbb N}
\newcommand{\Q}{\mathbb Q}
\newcommand{\R}{\mathbb R}
\newcommand{\D}{{\mathbb D}}
\newcommand{\Z}{{\mathbb Z}}
\newcommand{\A}{{\mathbb A}}
\newcommand{\B}{{\mathbb B}}
\newcommand{\wreath}[2]{#1 \otimes #2}

\newcommand{\Aa}{{\cal A}}
\newcommand{\Bb}{{\cal B}}
\newcommand{\Cc}{{\cal C}}
\newcommand{\Qq}{{\hat Q}}
\newcommand{\Hh}{{\cal H}}
\newcommand{\aaut}{{\cal A}}
\newcommand{\baut}{{\cal B}}
\newcommand{\caut}{{\cal C}}

\newcommand{\trans}[3]{ #1 \stackrel{#2}{\longrightarrow} #3 }
\newcommand{\trrule}[3]{(#1,#3)}
\newcommand{\transtrans}[3]{ #1 \stackrel{#2}{\longrightarrow}^* #3 }
\newcommand{\horiztrans}[3]{ #1 \stackrel{#2}{\leadsto} #3 }
\newcommand{\horiztranstrans}[3]{ #1 \stackrel{#2}{\leadsto}^* #3 }
\newcommand{\horiztransplus}[3]{ #1 \stackrel{#2}{\leadsto}^+ #3 }
\newcommand{\invreach}[2]{\text{Reach}^{-1}_{#1}(#2)}
\newcommand{\invacc}[1]{\text{Reach}^{-1}_{#1}}
\newcommand{\ew}{\varepsilon}
\newcommand{\eps}{\varepsilon}
\newcommand{\slcomm}[1]{\todo[color=blue!20]{SL: #1}}
\newcommand{\exptime}{{\tt ExpTime}\xspace}
\newcommand{\ptime}{{\tt PTime}\xspace}
\newcommand{\PDS}{PDS\xspace}
\newcommand{\pds}{\mathcal P}
\newcommand{\NFA}{NFA\xspace}
\newcommand{\nfa}{\mathcal A}
\newcommand{\nfb}{\mathcal B}
\newcommand{\push}{\mathsf{push}}
\newcommand{\pop}{\mathsf{pop}}
\newcommand{\nop}{\mathsf{nop}}
\newcommand{\forced}{\text{forced}}

\newcommand{\eqdef}{\stackrel {\text{def}} =}
\newcommand{\func}[3]{\mathop{\mathchoice{%
    {#1}\colon\;{#2}\;\longrightarrow\;{#3}}{%
    {#1}\colon{#2}\to{#3}}{%
    script}{%
    sscript}
  }}
\newcommand{\defin}[1]{[#1]}
\newcommand{\ov}[1]{\overline{#1}}
\newcommand{\aut}[1]{\text{Aut}(#1)}
\newcommand{\orbit}[1]{\text{orbit}(#1)}
\newcommand{\fodef}{FO-definable\xspace}
\newcommand{\Fodef}{FO-definable\xspace}
\newcommand{\qfdef}{qf-definable\xspace}
\newcommand{\Qfdef}{Qf-definable\xspace}
\newcommand{\prestar}[1]{\text{Pre}^*(#1)}
\newcommand{\autom}{{\cal A}}
\newcommand{\quot}[2]{#1/{#2}}
\newcommand{\set}[1]{\left\{ #1 \right\}}
\newcommand{\setof}[2]{\set{ #1 \; | \; #2 }}
\newcommand{\tuple}[1]{\langle #1 \rangle}
\newcommand{\Tuple}[1]{\left\langle #1 \right\rangle}
\newcommand{\goesto}[1]{\stackrel {#1} \longrightarrow}
\newcommand{\lang}[2]{\mathcal L_{#1}({#2})}
\newcommand{\poly}{\text{poly}}

\newcommand{\proj}[2]{#2|_{#1}}
\newcommand{\bigrule}[2]{\frac{#1}{#2}}
\renewcommand{\phi}{\varphi}
\newcommand{\vars}[1]{\text{vars}(#1)}

\newcommand{\rhopush}{\text{push}}
\newcommand{\rhopop}{\text{pop}}
\newcommand{\rhonop}{\text{nop}}
\newcommand{\concat}{\cdot}
\newcommand{\elt}[1]{\langle #1\rangle}

\newcommand{\ignore}[1]{}

%% file: intro.tex

\section{Introduction}

\subparagraph*{Context.}

Pushdown automata (\PDS) are a well-known model of recursive programs, with applications in areas
as diverse as language processing, data-flow analysis, security, computational biology, and program verification.
Many interesting analyses reduce to checking reachability in the infinite configuration graph generated by a \PDS,
which can be done in PTIME with the popular \emph{saturation algorithm}
\cite{BouajjaniEsparzaMaler:Pushdown:1997,FinkelWillemsWolper:Pushdown:1997}
(cf. also the recent survey \cite{CarayolHague:SaturationSurvey:2014}).
Saturation shows a slightly more general property of \PDS graphs,
which is sometimes called \emph{effective preservation of regularity}:
For a regular set of target configurations of a given \PDS,
the set of all configurations which can reach the target in a finite number of steps is effectively regular too.
The preservation is \emph{effective} in the sense that there exists a procedure which produces,
from an NFA recognizing the target set,
an NFA recognizing the predecessors.
This is a central theoretical result in the analysis of \PDS,
with immediate practical applications as demonstrated by the prominent tool MOPED \cite{Esparza:Schwoon:MOPED:2001}.
Therefore, it is of interest to extend this conceptually simple and yet powerful method to more general settings.

Several generalizations of the pushdown structure yielding \PDS-like models admitting effective preservation of regularity are known,
e.g., tree-pushdown systems \cite{Guessarian:TPDA:1983},
ordered multi-pushdown systems \cite{BreveglieriCherubiniCitriniCrespi-Reghizzi:Ordered:1996,Atig:ordered:2012},
annotated higher-order pushdown systems \cite{Maslov:Multilevel:1976,BroadbentCarayolHagueSerre:Saturation:2012},
and strongly normed multi-pushdown systems~\cite{CHL13lmcs}.
In this paper, instead of generalizing the pushdown structure itself,
we generalize the \emph{contents} of the pushdown,
by allowing the pushdown symbols to be drawn from an infinite set.
%
Our model is parametric in the choice of a countably-infinite logical structure $\atoms$, called \emph{atoms}.
%
%
%
We introduce and study \emph{first-order definable pushdown systems} (\fodef \PDS) over $\atoms$, which are like usual \PDS,
except that control locations, stack alphabet, and transition relation
are \fodef sets over $\atoms$, instead of ordinary finite sets.
Thus, we do not invent a new model,
but we reinterpret the classical model in a new setting.
This covers ordinary \PDS as a special case, and allows the study of non-trivial yet decidable classes of \PDS over infinite alphabets.
%
%
For instance, by taking $\atoms$ to be \emph{equality atoms} $(\D, =)$,
i.e., a countably-infinite set $\D$ where only equality testing is allowed,
we obtain (and slightly generalize) pushdown register automata \cite{ChengKaminski:CFL:AI98,BolligCyriacGastin:FOSSACS:2012,MRT14}.

\subparagraph*{Contributions and organization.}

The technical results of this paper and its structure are as follows.
%
%
%
	In Sec.~\ref{sec:fodef sets}, we recall the setting of \fodef sets, \fodef relations, and \fodef NFA.
	In Sec.~\ref{sec:fodef aut},
	we introduce \fodef \PDS.
	This is done by reinterpreting the classical model in the \fodef framework.
	Our approach has the advantage that we do not need to define a new model.
	Instead, we \emph{reinterpret} the classical model in a generic logical framework.
	In Sec.~\ref{sec:oligomorphic},
	we consider \emph{oligomorphic atoms}%
	\footnote{A structure $\atoms$ is \emph{oligomorphic} if for every $n$, the product $\atoms^n$ is orbit-finite.}
	with a decidable first-order theory,
	and we show effective preservation of regularity
	for the backward reachability relation of configuration graphs of \fodef \PDS.
	This is obtained via a symbolic implementation of the classical saturation method,
	which comes along with a simple proof of correctness.
	In Sec.~\ref{sec:homogeneous},
	we provide an upper complexity bound in the special case of \emph{homogeneous} atoms,
	and in particular an \exptime bound
	in the case of \emph{tractable} homogenous atoms,
	matching the known \exptime-hardness for equality atoms from \cite{MRT14}.
	%
	%
	In Sec.~\ref{sec:examples},
	we provide many interesting examples of tractable homogeneous 
	atoms for which we can apply our results,
	including equality atoms \cite{MRT14} (as remarked above),
	but also:
	\emph{total order atoms} $(\Q, \leq)$, which can be used for modeling densely-ordered data values;
	\emph{equivalence atoms} $(\D, R)$,
	where $R$ is an equivalence relation of infinite index s.t. each equivalence class is infinite,
	which can be used to model nested data values;
	\emph{universal tree atoms},
	which can be used to model dynamic topologies of concurrent programs with process creation and termination;
	as well as other structures,
	such as \emph{universal partial order atoms},
	\emph{universal tournament atoms},
	and \emph{universal graph atoms} \cite{survey}.
	In the same section, we also show that the classic \emph{wreath product} construction
	can be used to generate infinitely many new tractable examples from previous ones.
	%
%
Our logical approach has the advantage to highlight the general principle behind decidability,
and we can thus prove correctness once and for all for \emph{all} structures satisfying the mild assumptions above.
As a byproduct, we also obtain tight complexity results for \PDS over natural classes of infinite alphabets.
Infinitely many such natural structures can be found by using the wreath product construction.
%
%
%
%
In Sec.~\ref{sec:conclusions}, we conclude with some directions for future work.

%% file: prelim.tex

\section{Preliminaries}
\label{sec:fodef sets}

\mysubsection{Sets with atoms}

Let $\atoms$ be a countably-infinite logical structure with finite vocabulary.
An element of the structure we call \emph{atom}, and the whole structure we call \emph{atoms}. 
Examples of atoms are equality atoms $(\D, =)$, i.e., an arbitrary countable infinite set $\D$ with equality,
and total order atoms $(\Q, \leq)$, i.e., the rationals with the dense order.
More examples of atoms will be discussed in Sec.~\ref{sec:examples}.
%
%
In the study of atoms, the group $\aut{\atoms}$ of automorphisms\footnote{An \emph{automorphism} is a bijection of atoms that preserves all relations from the vocabulary.} of $\atoms$ plays a central role.
For instance, automorphisms of equality atoms are all permutations of $\D$,
and automorphisms of total order atoms are monotonic permutations of $\Q$.
By using atoms, we can build sets containing either previously built sets, or atoms themselves. 
For example, we build tuples $\atoms^n$ of fixed length, or disjoint unions thereof.
On such sets, we will consider the natural action of $\aut{\atoms}$,
which renames atoms while keeping intact the remaining structure.
For instance, on tuples of atoms the natural action is the point-wise renaming:
for $\pi \in \aut{\atoms}$ and $a_1, \ldots, a_n \in \atoms$,
$\pi(a_1, \ldots, a_n) \ = \ (\pi(a_1), \ldots, \pi(a_n))$.
Similarly, on disjoint unions the action is component-wise.
The action induces the notion of \emph{orbit},
which is the set of elements that can be reached via renaming,
i.e., $\orbit{e} = \setof{ \pi(e) } { \pi \in \aut{\atoms} }$.
The sets in the sequel will always be \emph{equivariant}, i.e., invariant under action of automorphisms\footnote{%
More generally, one can consider \emph{finitely supported} sets. A set is supported by $S \subseteq_\text{fin} \atoms$
if it is invariant under automorphisms that preserve elements of $S$. The results of this paper can be straightforwardly generalized
to finitely supported sets.}.
Every orbit is equivariant by definition, and every equivariant set is a disjoint union of orbits.
For instance, in total order atoms $(\Q, \leq)$,
the set $\Q^2$ is the disjoint union of 3 orbits,
$\setof{(q,q')}{q < q'}$, $\setof{(q,q')}{q = q'}$, and $\setof{(q,q')}{q > q'}$;
and $\Q^2 \uplus \Q^3$ is the disjoint union of 16 orbits.
A central notion 
is that of \emph{orbit-finite} sets, which are \emph{finite} unions of orbits (as opposed to arbitrary unions).
Intuitively, an orbit-finite set has only finitely many elements up to renaming by atom automorphisms.
Orbit-finiteness generalizes finiteness,
and a substantial portion of results from automata theory carry over to the more general orbit-finite setting~\cite{BKL11full}.
This paper can be seen as such a case study for the specific case of pushdown automata.
For the sake of concreteness, we restrict in the rest of the paper to \emph{\fodef sets}, to be defined now; 
we only note that the results of this paper can be straightforwardly generalized to all orbit-finite sets with atoms.

\ignore{

...Every orbit is a quotient of an orbit of tuples of atoms, by an equivariant equivalence.
Equivariant equivalence is also a set of tuples, hence FO-definable. This defines representations
of arbitrary orbit-finite sets.

Thus every orbit X has an equivariant projection O --> X, from some orbit O of tuples of atoms.
For reachability analysis, one may consider O instead of X, i.e. consider inverse image of
states, alphabets, transitions, etc along the projections mentioned above. This reduces the 
general case of all orbit-finite sets to FO-definable sets.

}

\mysubsection{\Fodef sets}

Fix a structure $\atoms$ over a finite vocabulary.
We describe infinite sets symbolically using first-order logic over the vocabulary of $\atoms$,
which we assume to always include the equality relation $=$.
A first-order formula $\phi(\vec x)$
(where we explicit list all free variables according to an implicit order)
with $n \geq 1$ free variables
\emph{defines} the subset $\defin{\phi} \subseteq \atoms^n$
of tuples that satisfy $\phi$, i.e.,
$\defin{\phi} = \setof{ \vec a \in \atoms^n}{(\vec x \mapsto \vec a) \vDash \phi}$.
This set is always equivariant,
since a formula can only compare atoms by using symbols from the signature,
and automorphisms by definition respect this signature.
The \emph{dimension} of $\defin{\phi}$ is the number $n \geq 1$ of free variables of $\phi$,
denoted by $\dim \phi$. 
We also allow the tautologically true formula $\phi \equiv (\forall x \cdot x = x)$;
by convention, we take $\dim\phi = 0$ and $\defin\phi$ is a singleton (for a fixed atom in $\atoms$).
%
%
A \emph{\fodef set $X$ over $\atoms$} is a finite indexed union of such sets,
i.e., 
\[ X = \bigcup_{l \in L} \set l \times \defin{\phi_l} , \qquad \text{ where $L$ is a finite index set.} \]
When we want to omit the formal indexing, we just write $X$ as the finite disjoint union $\biguplus_{l \in L} \defin{\phi_l}$.
Since \fodef sets are unions of equivariant sets, they are equivariant too.
When $\dim {\phi_l} = 0$ for every $l \in L$, then $X$ is finite and has the same number of elements as $L$.
Thus, \fodef sets generalize finite sets.

We use \fodef sets for control locations and alphabets of automata.
In the former case, an index $l \in L$ may be understood as a control location,
and a tuple $\vec{a} \in \atoms^n$ as a valuation of $n$ registers.
Under this intuition, $\phi_l$ is an invariant
that constrains register valuations in a control location $l$.
We do not assume that all component sets $\defin{\phi_l}$ have the same dimension,
i.e., the number of registers may vary from one control location to another.

%
%

\mysubsection{\Fodef relations}

Along the same lines, we define \fodef binary relations.
Consider two \fodef sets $X = \biguplus_{l \in L} \defin{\phi_l}$ and $Y = \biguplus_{k \in K} \defin{\psi_k}$.
An \fodef relation $R \subseteq X \times Y$ is an \fodef set
$R = \biguplus_{l \in L, k \in K} \defin{\xi_{l k}}$
where the indexing set is the Cartesian product $L \times K$,
and every component set $\defin{\xi_{l k}}$ satisfies
$\defin{\xi_{l k}} \subseteq \defin{\phi_l} \times \defin{\psi_k}$.
In particular, $\dim {\xi_{l k}} = \dim {\phi_l} + \dim {\psi_k}$.
Relations of greater arities can be obtained by iterating the construction above.
We use \fodef relations to define transition relations of automata.
The formula $\xi_{l k}$ may be understood as a constraint on a transition from control location $l$ to control location $k$,
prescribing how a valuation of registers in $l$ before the transition relates to a valuation of registers in $k$ after the transition.


\mysubsection{\fodef \NFA}

As an example application of \fodef sets and relations, we define \fodef \NFA.
This model will be used later to recognize regular set of configurations of \fodef \PDS, also defined later.
A classical \NFA is a tuple $\nfa = (\Gamma, Q, F, \delta)$,
where $\Gamma$ is a finite input alphabet,
$Q$ is a finite set of states, of which those in $F \subseteq Q$ are the final ones,
and $\delta \subseteq Q \times \Gamma \times Q$ is the transition relation.
Once an initial state is chosen,
the definitions of run, accepting run, and language $\lang {} \nfa$ recognized by $\nfa$ are standard.
By simply replacing ``finite'' with ``\fodef'' in the definition above, we obtain \fodef \NFA.
To fix notation, an \fodef \NFA will be written as a tuple
$\nfa = (
	\Gamma = \biguplus_{k \in K} [\phi_k], \ 
	Q = \biguplus_{l \in L} [\psi_l], \ 
	F = \biguplus_{l \in L} [\psi^F_l], \ 
	\delta = \biguplus_{l, l' \in L, k \in L} [\delta_{lkl'}])$,
where w.l.o.g. we assume that $Q$ and $F$ have the same index set $L$.
Notice that $\delta$ is an \fodef set, while $\delta_{lkl'}$ is a first-order formula.

\begin{example}
	\label{ex:NFA}
	Let $\atoms$ be the total order atoms $(\Q, \leq)$,
	and let the alphabet be $\Gamma = \set{k} \times \Q$.
	Consider the language
	$M = \setof{ (k, a_1) \cdots (k, a_n) \in \Gamma^*}{ a_1 \geq a_2 \leq a_3 \geq \cdots \leq a_{2n+1}}$ 
	of non-empty finite words of odd length of alternating growth.
	This language 
	can be recognized from state $\ell_I$ by the \NFA
	$$\nfa = (
		\Gamma, \ 
		Q = \set{\ell_I} \cup \set{\ell_0} \times \Q \cup \set{\ell_1} \times \Q,
		F = \set{\ell_0} \times \Q,
		\delta = \biguplus_{l, l' \in \set {\ell_I, \ell_0, \ell_1}} [\delta_{lkl'}]
		)
		.$$
	The initial location $\ell_I$ does not contain any register,
	while control locations $\ell_0, \ell_1$ both contain one register,
	which is used to guess the next input symbol and to ensure the right ordering.
	Formally, 
	$\delta_{\ell_I k \ell_0}(, y, x') \equiv x' \leq y$
	(we use the notation $\delta_{\ell_I k \ell_0}(, y, x')$ to emphasize that $\ell_I$ does not have any register),
	$\delta_{\ell_0 k \ell_1}(x, y, x') \equiv (x = y \wedge x' \geq y)$,
	$\delta_{\ell_1 k \ell_0}(x, y, x') \equiv (x = y \wedge x' \leq y)$,
	and $[\delta_{l k l'}] = \emptyset$ for the other cases.
\ignore{

	For the total order atoms, consider an alphabet $A$ and a set of states $Q$, indexed by $\{a, b\}$ and $\{p, q\}$, respectively: 
	\begin{align*}
	A_a \ & =  \  \Q & Q_p \ & = \  \{ (t, t') \in \Q^2 \ : \ t \leq t' \} \\
	A_b \ & =  \  \Q^2 & Q_q \ & = \  \Q
	\end{align*}
	Thus in control state $p$ there are two registers, with the value of the first one smaller or equal to the value of the second one; 
	and in control state $q$ there is just one register.

	To define a ternary \fodef relation $\delta \subseteq Q \times A \times Q$ we use the indexes
	$w \in \{p, q\} \times \{a,b\} \times \{p, q\}$. 
	In order to avoid introducing unnecessary entities, we write $\delta_{w}$ for a basic \fodef set as well as for the formula defining
	this set, hoping that this identification will not lead to confusion. 
	Below, variables $x$ or $\vec{x} = (x_1, x_2)$ range over $Q$, and variables $y$ or $\vec{y} = (y_1, y_2)$ 
	range over $A$. 
	Note that the number of free variables of $\delta_w$ depends on the index $w$.

	\begin{align*}
	\delta_{p a p}(\vec{x}, y, \vec{x'}): & \quad x_1 \leq y \leq x_2 \ \land \ x'_1 = y \ \land \ x'_2 = x_2 \\
	\delta_{p b q}(\vec{x}, \vec{y}, x'): & \quad x_1 = y_1 \ \land \ x_2 = y_2 \ \land \  (x' = x_1 \lor  x' = x_2) \\
	\delta_{q b p}(x, \vec{y}, \vec{x'}): & \quad x = y_1 \ \land \ x'_1 = x'_2 = y \\
	... & \quad \text{ (other cases omitted) }
	\end{align*}

	The relation $\delta$ describes $A$-labeled transitions between states from $Q$.
	For instance, in state $p(x_1, x_2)$ (which means control state $p$ with register values $x_1, x_2$), 
	if $a(y)$ is read, a transition may go to state $p(x'_1, x'_2)$ if the first formula above holds, i.e. the value $y$ read from input
	is between the values of registers, and the value of the first register is updated to $y$. 
	Likewise, if $b(y_1, y_2)$ is read, a transition from control state $p$ to $q$ is possible assuming that the input values
	$y_1, y_2$ are equal to register values; the value of register in control state $q$ is chosen nondeterministically among
	the two input values. Similarly, upon reading $b(y_1, y_2)$ a transition is possible from control state $q$ back to $p$
	if the register value is equal to $y_1$; the values of registers in $p$ are both set to $y_2$. 
}
\end{example}

%% file: fodef-aut.tex

\section{First-order definable pushdown systems}
\label{sec:fodef aut}

In this section we define \fodef \PDS and their reachability problem.
According to the classical definition, a pushdown system (\PDS) $\pds = \tuple {\Gamma, P, \rho}$
consists of a finite stack alphabet $\Gamma$, 
a finite set of control states $P$,
and a finite set of transition rules $\rho = \rho^\push \cup \rho^\pop$,
which is partitioned into push rules $\rho^\push \subseteq P \times \Gamma \times P \times \Gamma \times \Gamma$
and pop rules $\rho^\pop \subseteq P \times \Gamma \times P$.
%
%
In this paper, we reinterpret this definition in the setting of \fodef sets,
which yields a more general model.
For an atom structure $\atoms$,
\emph{\fodef \PDS over $\atoms$} are obtained by replacing ``finite set'' with ``\fodef set'' in the classical definition.
To fix notation, an \fodef \PDS is a tuple
\[\pds = \tuple {
	\Gamma = 
	\biguplus_{k \in K} [\phi_k], \
	P = \biguplus_{\ell \in L} [\xi_\ell], \ 
	\rho = \rho^\push \cup \rho^\pop}, \] 
where%
\footnote{We could have also considered push rules which do not read the top of the stack,
i.e., of the form
$\rho^\push = \biguplus_{\ell, \ell' \in L, k' \in K} [\rho^\push_{\ell \ell' k'}]$.
However, these would introduce $\epsilon$-transitions during our saturation procedure in Sec.~\ref{sec:oligomorphic},
which we want to avoid for simplicity.}
$\rho^\push = \biguplus_{\ell, \ell' \in L, k, k', k'' \in K} [\rho^\push_{\ell k \ell' k' k''}]$
and $\rho^\pop = \biguplus_{\ell, \ell' \in L, k \in K} [\rho^\pop_{\ell k \ell'}]$.
As in the classical case, an \fodef \PDS induces an infinite transition system $\tuple{\mathcal C, \trans {} {} {}}$,
where the set of configurations is $\mathcal C = P \times \Gamma^*$,
and there is a transition $\trans c {} {c'}$ between two configurations $c = (q, a w)$ and $c' = (q', w')$
if, and only if, either there exists a push rule $(q, a, q', b, c) \in \rho^\push$ s.t. $w' = b c w$,
or there exists a pop rule $(q, a, q') \in \rho^\pop$ s.t. $w = w'$.
Let $\transtrans{}{}{}$ be the reflexive and transitive closure of $\trans {} {} {}$.
For a set $C$ of configurations, the \emph{backward reachability set of $C$}, denoted $\invreach{\pds}{C}$,
is the set of configurations that can reach some configuration in $C$:
\begin{align*}
	\invreach{\pds}{C} = \setof{ c \in \mathcal C}{\transtrans{c}{}{c'} \text{ for some } c' \in C} \ .
\end{align*}

\ignore{
\begin{remark}
	A seemingly more general version of \PDS can be considered by allowing arbitrary prefix rewriting in $\pds$.
	In this case, $\rho$ is an \fodef subset of $P \times \Gamma^* \times P \times \Gamma^*$.
	We argue that, like in the classical case,
	there is no loss of generality in restriction to simple push and pop operations
	when one is only concerned with the reachability set.
	Since $\rho$ is \fodef, there exists $n \geq 0$
	s.t. 
	$\rho \subseteq P \times \Gamma^{\leq n} \times P \times \Gamma^{\leq n}$,
	where $\Gamma^{\leq n} = \set{\varepsilon} \cup \Gamma^1 \cup \Gamma^2 \cup \cdots \cup \Gamma^n$.
	Let $\pds' = \tuple {\Gamma, P', \rho'}$.
	The new set of states is $P' = P \times \Gamma^{\leq n}$,
	which is an \fodef set, since $P$ and $\Gamma$ are \fodef.
	Indeed, if $P = \biguplus_{\ell \in L}[\xi_\ell]$ and $\Gamma = \biguplus_{k \in K} [\phi_k]$,
	then we can represent $P'$ as
	$P' = \biguplus_{\ell \in L, \vec k \in K^{\leq n}}[\psi_{\ell \vec k}]$,
	where, for every $\ell k_1 \dots k_m \in L\times K^{\leq n}$,
	\begin{align*}
		\psi_{\ell k_1 \cdots k_m}(\vec x, \vec y_1, \dots, \vec y_m) \ \equiv \ 
			\xi_\ell(\vec x) \wedge \phi_{k_1}(\vec y_1) \wedge \cdots \wedge \phi_{k_m}(\vec y_m) \ .
	\end{align*}
	It remains to define transitions in $\rho' = \rho^\push \cup \rho^\pop \cup \rho^\nop$.
	For simplicity, we additionally use transitions $\rho^\nop \subseteq P\times P$ that do not change the top of the stack.
	They can easily be replaced by a dummy push followed by a pop.
	A long transition $(q, u, q', v) \in \rho$ is simulated by $((q, u), (q', v')) \in \rho^\nop$.
	In order to unload/load the local buffer,
	we add push/pop transitions $((q, a w), (q, w), a) \in \rho^\push$ 
	and $((q, w), a, (q, a w)) \in \rho^\pop$ for every $q \in P$ and $a w\in \Gamma^{\leq n}$.
	At the level of formulas (recall that $\rho^\nop_{\dots}, \rho^\push_{\dots}, \rho^\pop_{\dots}$ are formulas),
	{\small
	\begin{align*} 
		&\rho^\nop_{(\ell k_1 \cdots k_m) (\ell' k'_1 \cdots k'_{n})}
			(\vec x \vec y_1 \dots \vec y_{k_m}, \ \vec x' \vec y_1' \dots \vec y_{k'_{n}}) \equiv
				\rho_{(\ell k_1 \cdots k_m ) (\ell' k'_1 \cdots k'_{n})}
					(\vec x, \vec y_1, \dots, \vec y_{k_m}, \vec x', \vec y_1', \dots, \vec y_{k'_{n}})
		\\
		&\rho^\push_{(\ell k_0 k_1 \cdots k_m)(\ell k_1 \cdots k_m) k_0}
			(\vec x \vec y_0 \vec y_1 \dots \vec y_m, \ \vec x' \vec y_1' \dots \vec y_m', \ \vec y_0') \equiv
				\vec x' = \vec x \wedge \vec y_0 = \vec y_0' \wedge \cdots \wedge \vec y_m' = \vec y_m
		\\
		&\rho^\pop_{(\ell k_1 \cdots k_m) k_0 (\ell k_0 k_1 \cdots k_m)}
			(\vec x \vec y_1 \dots \vec y_m, \ \vec y_0, \  \vec x' \vec y_0' \vec y_1' \dots \vec y_m') \equiv
				\vec x' = \vec x \wedge \vec y_0 = \vec y_0' \wedge \cdots \wedge \vec y_m' = \vec y_m
	\end{align*}
	}
	This translation preserves reachable configurations, in the following sense.
	Let $C \subseteq  P \times \Gamma^*$ be a set of configurations of $\pds$.
	Clearly, $C \subseteq P' \times \Gamma^*$ is also a set of configurations of $\pds'$,
	and $\invreach{\pds}{C} = \invreach{\pds'}{C} \cap (P \times \Gamma^*)$.
	Therefore, w.l.o.g. in the following we consider \PDS with only push and pop transitions.
\end{remark}
}

\begin{example}
	\label{ex:PDS}
	We define an \fodef \PDS $\pds$ over total order atoms $(\Q, \leq)$ which constructs strictly monotonic stacks,
	the maximal element being on the top of the stack.
	Let $\pds = \tuple {
		\Gamma = \set k \times \Q,
		P = \set {\ell_I},
		\rho = \rho^\push}$,
	where 
	$\rho^\push_{\ell_I k \ell_I k k}(, y, , y', y'') \equiv (y < y' \wedge y'' = y)$.
	%
	%
\end{example}

This paper concentrates on the reachability analysis for \fodef \PDS.
Given an \fodef \PDS $\pds = \tuple{\Gamma, P, \rho}$, two control locations $p, q \in P$, and a stack symbol $\bot \in \Gamma$,
the \emph{reachability problem} asks whether $(p, \bot) \in \invreach{\pds}{\set q \times \Gamma^*}$. 
We start with stack $\bot$ and we ignore the stack at the end of the computation.
More general analyses can be considered by imposing regular constraints on the initial and final stack contents.
These easily reduce to reachability of a regular set of configurations, which is the problem considered in the next section.

%% file: saturation.tex

\section{Preservation of regularity I: Oligomorphic atoms}

\label{sec:oligomorphic}

We solve the reachability problem as a corollary of a general effective preservation of regularity result
for the backward reachability relation of \fodef \PDS.
To this end, we use \fodef \NFA to describe regular sets of configurations.
In the following, fix an \fodef \PDS $\pds = \tuple {\Gamma, P, \rho}$,
and an \fodef \NFA $\nfa = \tuple {\Gamma, Q, F, \delta}$ s.t. $P \subseteq Q$.
The NFA $\nfa$ recognizes the following language $L_{\pds}(\nfa)$ of configurations of $\pds$,
\begin{align*}
	\lang \pds \nfa = \setof{ (p, w) \in P \times \Gamma^*} {\nfa \text{ accepts } w \text{ from state } p }.
\end{align*}
Such sets of configurations of $\pds$ we call \emph{regular}.
We assume w.l.o.g.~that states of $\nfa$ that belong to $P$ do not have incoming transitions, 
i.e.~$\delta \subseteq Q \times \Gamma\times (Q \setminus P)$.

\begin{example}
	\label{ex:predecessors}
	Recall the \fodef \PDS $\pds$
	from Example~\ref{ex:PDS} building strictly monotonic stacks
	(maximal element on top).
	Let $N$ be the following set of configurations
	$$N = \setof{ (\ell_I, (k, a_1) \cdots (k, a_{2n+1})) \in P \times \Gamma^*} {a_1 \geq a_2 \leq a_3 \geq \cdots \leq a_{2n+1}}.$$
	This set is regular, and it is 
	recognized by the \NFA $\nfa$
	from Example~\ref{ex:NFA},
	i.e., $\lang \pds \nfa = N$.
	The backward reachability set is
	$$\invreach{\pds}{N} = N \cup \setof
			{ (\ell_I, (k, a_2) \cdots (k, a_{2n+1})) \in P \times \Gamma^*}
			{ a_2 \leq a_3 \geq \cdots \leq a_{2n+1}}.$$
%
	%
	We will see below how to compute an \fodef \NFA recognizing $\invreach{\pds}{N}$.
\end{example}

We solve the reachability problem for \PDS over \emph{oligomorphic} atoms.%
\footnote{%
	One could also consider \PDS defined by general prefix rewriting,
	i.e., with transitions in $\rho \subseteq P \times \Gamma^* \times P \times \Gamma^*$.
	For oligomorphic atoms, our simplified push/pop model can simulate prefix rewriting
	while preserving reachability properties (but not configuration graph isomorphism, or even bisimilarity),
	like in the classical case.
}.
Oligomorphicity is an important notion in model theory \cite{survey}.
Formally, a structure is oligomorphic if, and only if, for every $n \in \nat$, the set $\atoms^n$ is orbit-finite.
Not all structures are oligomorphic, as shown in the following example.
%
%
\begin{remark}[Timed atoms]
	Timed atoms $(\Q, \leq, +1)$ is a well-known example of non-oligomorphic structure.
	They extend total order atoms $(\Q, \leq)$ with the successor relation $(+1) \subseteq \Q \times \Q$.
	Automorphisms of timed atoms are monotone bijections $\pi$ of $\Q$ that preserve unit intervals, i.e., $\pi(x + 1) = \pi(x) + 1$.
	To see why timed atoms are non oligomorphic, 
	it suffices to see that already $\Q^2$ has infinitely-many orbits.
	Indeed, for each $z \in \Z$, $\Q^2$ has a disjoint orbit $\setof{(x, y) \in \Q^2}{x - y = z}$.
	(Since automorphisms preserve unit intervals, they preserve all integer distances.)
	Working in non-oligomorphic structures like timed atoms requires the use of specialized techniques,
	and the generic algorithm presented in this section does not terminate.
	We have thoroughly studied the reachability problem for \fodef pushdown systems and automata over timed atoms in \cite{trPDA}.
\end{remark}

Since oligomorphic atoms are very general, we can merely state decidability of the reachability problem, without any complexity bounds. 
The only additional assumption that we require is decidability of the \emph{first-order satisfiability problem} in the structure $\atoms$, which asks, given a first-order formula $\phi(x_1, \ldots, x_n)$,
whether some valuation $\eta:\{x_1, \ldots, x_n\} \to \atoms$ of its free variables satisfies $\phi$.
\begin{theorem} \label{thm:pda decid oligo}
	Let $\atoms$ be an oligomorphic structure with a decidable first-order satisfiability problem.
	For \fodef \PDS $\pds$ over $\atoms$ and an \fodef \NFA $\nfa$ over $\atoms$ recognizing a regular set of configurations $L_{\pds}(\nfa)$,
	one can effectively construct an \fodef \NFA $\nfb$ over $\atoms$ recognizing $L_{\pds}(\nfb) = \invreach{\pds}{L_{\pds}(\nfa)}$.
\end{theorem}
We prove Theorem~\ref{thm:pda decid oligo} by using the classical \emph{saturation technique} \cite{BouajjaniEsparzaMaler:Pushdown:1997,FinkelWillemsWolper:Pushdown:1997}.
We first describe a simple abstract algorithm manipulating infinite sets of transitions,
and then we show how this can be implemented symbolically at the level of formulas.
As in the classical case, the \fodef \NFA $\nfb$ which is computed by the algorithm 
is of the form $\tuple {\Gamma, Q, F, \delta'}$ with $\delta \subseteq \delta'$,
i.e., it is obtained by adding certain transitions to $\nfa$.
For any relation $\alpha \subseteq Q \times \Gamma \times Q$,
let $\forced(\alpha) \subseteq Q \times \Gamma \times Q$ be the following set of triples:
\begin{align*}
	\forced(\alpha) = \setof{(q, a, q')} 
		{\exists \trrule{q, a}{}{q'', b, c} \in \rho^\push, \exists (q'', b, q''') \in \alpha, \exists (q''', c, q') \in \alpha}.
\end{align*}
The abstract saturation algorithm is shown in Fig.~\ref{fig:abstract:saturation}.
\begin{figure}
	\begin{align*}
	(0) \ \ & \delta' \ := \ \delta \ \cup \ \rho^\text{pop} \\
	(1) \ \ & \text{\tt repeat} \\
	(2) \ \ & \qquad \delta' \ := \ \delta' \ \cup \ \forced(\delta') \\
	(3) \ \ & \text{\tt until} \text{ forced}(\delta') \subseteq \delta'
	\end{align*}
	\caption{Abstract saturation algorithm.}
	\label{fig:abstract:saturation}
\end{figure}
The algorithm is partially correct for every structure $\atoms$ (even though it might not terminate).
This follows directly from the observation that the saturated \NFA $\nfb$
has a transition $(q, a, q') \in \delta'$ between states $q, q' \in P$ of $\pds$
if, and only if, $\pds$ admits a run $\transtrans{(q, a)}{}{(q', \ew)}$ (we use here the assumption that no transition of $\nfa$ ends in
a state $q\in P$ of $\pds$).
However, on arbitrary structures saturation does not terminate,
either because the inclusion checking on line $(3)$ is not decidable,
or because it never actually holds.
The first issue is addressed by the requirement that $\atoms$ has a decidable first-order satisfiability problem,
and the second one by the fact that $\atoms$ is an oligomorphic structure.

We implement the abstract algorithm from Fig.~\ref{fig:abstract:saturation} symbolically,
by manipulating formulas instead of actual transitions.
We assume w.l.o.g.~that the index set of $P$ (the control locations of $\pds$) is the same as the index set of $Q$ (the states of $\nfa$).
First, notice that the set $\forced(\alpha)$ is \fodef whenever $\alpha$ is so, since it can be expressed as follows:
\begin{align*}
	\forced(\alpha)_{\ell k \ell'}(\vec x, \vec y, \vec x') := \!\!\!\!\!\!\!\!
		\bigvee_{\ell'', \ell''' \in L, k', k'' \in K} \!\!\!\!\!\!\!\!
			&\exists \vec x'', \vec y', \vec y'', \vec x''' \cdot
				\rho^\push_{\ell k \ell'' k' k''}(\vec x, \vec y,  \vec x'', \vec y', \vec y'')\ \wedge\ \\
					&\alpha_{\ell'' k' \ell'''}(\vec x'', \vec y', \vec x''')\ \wedge\ 
					\alpha_{\ell''' k'' \ell'}(\vec x''', \vec y'', \vec x'),
\end{align*}
where $L$ is the index set of $Q$, and $K$ is the index set of $\Gamma$.
Steps (0) (initialization of $\delta'$) and (2) (update of $\delta'$) of the algorithm are implemented by disjunction of \fodef sets,
therefore at each stage of the algorithm $\delta'$ is an \fodef set, and thus an equivariant set (i.e, a union of orbits).
The test (3) is computable whenever first order satisfiability is so.
We obtain the concrete algorithm in Fig.~\ref{fig:concrete:saturation}.
Termination is guaranteed since $\atoms$ is oligomorphic,
which implies orbit-finiteness of $Q \times \Gamma \times Q$.
Indeed, $\delta'$ is always a union of orbits at every stage,
and therefore at least one orbit is added to $\delta'$ at every iteration.
%


\begin{figure}
	\begin{align*}
		\textsf{INPUT:} &\ \textrm{an \fodef \PDS\ }
			\pds = \tuple {\Gamma = \biguplus_{k} \defin{\varphi_k}, 
				P = \biguplus_{\ell} \defin{\xi_\ell},  
				\rho^\push \cup \rho^\pop} \textrm{, with } \\
				&	\ 		\rho^\push = \biguplus_{\ell k \ell' k' k''} \defin{\rho^\push_{\ell k \ell' k' k''}}, 
							\rho^\pop = \biguplus_{\ell k \ell'} \defin{\rho^\pop_{\ell k \ell'}} 
							\textrm{, and an \fodef \NFA\ } \\
			& \nfa = \tuple {\Gamma,
				Q = \biguplus_{\ell} \defin{\psi_\ell},
				\delta = \biguplus_{\ell k \ell'} \defin{\delta_{\ell k \ell'}}}, \textrm{ with } 
				\defin{\xi_\ell} \subseteq \defin{\psi_\ell}, \textrm{ for every } \ell \in L. 		
	\end{align*}
	\begin{align*} 
	(0) \ \ & \textrm{for every } \ell, k, \ell':\  
		\delta'_{\ell k \ell'}(\vec x, \vec y, \vec x') \ := \
			\delta_{\ell k \ell'}(\vec x, \vec y, \vec x') \ \vee \ \rho^\pop_{\ell k \ell'}(\vec x, \vec y, \vec x') \\
	(1) \ \ & \text{\tt repeat} \\
	(2) \ \ & \qquad \textrm{for every } \ell, k, \ell':\    
				\delta'_{\ell k \ell'}(\vec x, \vec y, \vec x') \ := \ 
					\delta'_{\ell k \ell'}(\vec x, \vec y, \vec x') \ \vee \ 
					\forced(\delta')_{\ell k \ell'}(\vec x, \vec y, \vec x') \\ 
	(3) \ \ & \text{\tt until} (\bigwedge_{\ell, k, \ell'} \forall \vec x, \vec y, \vec x' \cdot
		\text{forced}(\delta')_{\ell k \ell'}(\vec x, \vec y, \vec x') \implies
			\delta'_{\ell k \ell'}(\vec x, \vec y, \vec x'))
	\end{align*}
	\caption{Concrete saturation algorithm; $\ell, \ell'$ range over $L$, and $k$ ranges over $K$.}
	\label{fig:concrete:saturation}
\end{figure}

\begin{example}
	We apply the concrete saturation algorithm to the \PDS $\pds$ and \NFA $\nfa$ from Example~\ref{ex:predecessors}.
	Recall that	$\pds = \tuple {\Gamma = \set k \cup \Q, P = \set {\ell_I}, \rho^\push}$,
	with $\rho^\push_{\ell_I k \ell_I k k}(, y, , y', y'') \equiv (y < y' \wedge y'' = y)$,
	and $\nfa = \tuple{ \Gamma, Q = \set {\ell_I} \cup \set{\ell_0, \ell_1} \times \Q, F = \set{\ell_0} \times \Q, \delta}$,
	with
	$\delta_{\ell_I k \ell_0}(, y, x') \equiv x' \leq y$,
	$\delta_{\ell_0 k \ell_1}(x, y, x') \equiv (x = y \wedge x' \geq y)$,
	$\delta_{\ell_1 k \ell_0}(x, y, x') \equiv (x = y \wedge x' \leq y)$
	(omitting the trivial cases).
	For the first iteration, let $\delta^0 := \delta$.
	We compute $\text{forced}(\delta^0)$,
	for which the only nontrivial case is
	$\text{forced}(\delta^0)_{\ell_I k \ell_1}(, y, x') \equiv \exists y', y'', x''' \cdot
		\rho^\push_{\ell_I k \ell_I k k}(, y, , y', y'') \wedge
		\delta^0_{\ell_I k \ell_0}(, y', x''') \wedge
		\delta^0_{\ell_0 k \ell_1}(x''', y'', x')$,
	which equals 
	$$	\exists y', y'', x''' \cdot (y < y' \wedge y'' = y) \wedge
		(x''' \leq y') \wedge
		(x''' = y'' \wedge x' \geq y'').$$
	By removing quantifiers (thanks to the density of $\Q$),
	the former is equivalent to $x' \geq y$.
	Therefore, $\delta^1$ extends $\delta^0$
	with the new transition $\delta^1_{\ell_I k \ell_1}(, y, x') \equiv (x' \geq y)$.
	Since $\delta^1$ is not equivalent to $\delta^0$, we go to the next iteration.
	We compute $\text{forced}(\delta^1)$,
	for which the only new case is
	$\text{forced}(\delta^1)_{\ell_I k \ell_0}(, y, x') \equiv \exists y', y'', x''' \cdot
		\rho^\push_{\ell_I k \ell_I k k}(, y, , y', y'') \wedge
		\delta^1_{\ell_I k \ell_1}(, y', x''') \wedge
		\delta^1_{\ell_1 k \ell_0}(x''', y'', x')$,
	which equals
	$$\exists y', y'', x''' \cdot
		(y < y' \wedge y'' = y) \wedge (x''' \geq y') \wedge (x''' = y'' \wedge x' \leq y'').$$
	The latter is equivalent to $\exists y' \cdot y < y' \wedge y \geq y' \wedge x' \leq y$,
	which is clearly unsatisfiable.
	Therefore $\delta^2$ is equivalent to $\delta^1$,
	and the algorithms stops.
	It is immediate to check that
	$\nfb = \tuple{\Gamma, Q = \ell_I \cup \set{\ell_0, \ell_1} \times \Q, F = \set{\ell_0} \times \Q, \delta^1}$
	recognizes precisely $\invreach{\pds}{N}$, where $N = \lang \pds \nfa$.
\end{example}

%% file: homogeneous.tex

\section{Preservation of regularity II: Homogeneous atoms}

\label{sec:homogeneous}

Relational homogeneous structures are a well-behaved subclass of oligomorphic structures,
for which we are able to give precise complexity upper bounds for our saturation construction.
A relational structure $\atoms$ (i.e., with no function symbols in the vocabulary) is \emph{homogeneous} if
every isomorphism between two finite induced substructures%
\footnote{An \emph{induced substructure} is a structure obtained by restricting the universe to a subset of atoms.}%
of $\atoms$ extends to an automorphism of the whole $\atoms$.
%
This immediately implies that $\atoms$ is oligomorphic.
\begin{proposition} \label{claim:nr of orbits}
	Let $\atoms$ be a relational homogeneous structure.
	For $n \geq 1$,
	the number of orbits of $\atoms^n$ is bounded by $2^{\poly(n)}$.
\end{proposition}
\begin{proof}
	A tuple of $n$ elements $(a_1, \dots, a_n) \in \atoms^n$
	can be seen as an induced substructure of $\atoms$,
	where elements are additionally labelled with the positions $\{1 \ldots n\}$.
	Two such induced substructures $\bar a, \bar b \in \atoms^n$ are isomorphic exactly when the elements $\bar a$ and $\bar b$
	satisfy the same relations in the vocabulary of $\atoms$.
	Therefore, there number of isomorphism classes is bounded by $2^{\poly(n)}$.
	Since $\atoms$ is homogeneous, every isomorphism between $\bar a$ and $\bar b$ extends to an automorphism of the whole $\atoms$,
	and thus $\bar a$ and $\bar b$ are in the same orbit.
	Consequently, the same bound applies to the number of orbits of $\atoms^n$.
\end{proof}

\noindent
All structures listed in the introduction are homogeneous relational structures.
However, not all oligomorphic relational structures are homogeneous as the example below shows.

\begin{example}[Bit vector atoms]
	Let a \emph{bit vector} be any infinite sequence of zeros and ones with only finitely many ones.
	A bit vector can be represented	by a finite sequence, by cutting off the infinite zero suffix.
	Consider the relational structure $\vatoms = (V, 0, +)$, consisting of the set $V$ of all bit vectors, together with
	a unary predicate $0(\_)$ that distinguishes the zero vector, and the ternary
	relation $\_+ \_ = \_$ that describes point-wise addition modulo 2.
	Automorphisms of $\vatoms$ are precisely linear mappings,
	i.e., bijections $f$ s.t. $f(0) = 0$ and $f(u + v) = f(u) + f(v)$.
	The orbit of a tuple $(v_1, \ldots, v_n) \in V^n$ is determined by its \emph{addition type},
	i.e., by the the set of all equalities of the form
	$v_{i_1} + \ldots + v_{i_m} = 0$ satisfied by $(v_1, \ldots, v_n)$.
	Indeed, for two tuples $(u_1, \ldots, u_n), (v_1, \ldots, v_n) \in V^n$ having the same addition type,
	consider the partial bijection $f$ defined as $f(u_1) = v_1, \dots, f(u_n) = v_n$.
	By using the Steinitz exchange lemma, the function $f$ can be extended to a linear mapping on the whole $V$,
	and thus $(u_1, \ldots, u_n)$ and $(v_1, \ldots, v_n)$ are in the same orbit.
	Therefore, the number of orbits of $V^n$ is finite.
	On the other hand, $\vatoms$ is not homogeneous.
	For instance, the two induced substructures
	$X = \set{1000, 0100, 0010, 0001}$ and $Y = \set{1000, 0100, 0010, 1110}$ are isomorphic.
	Define, e.g., $f(0001) = 1110$, and $f(x) = x$ if $x\neq 0001$.
	The reason why $f$ is an isomorphism is that $f$ needs to respect $\_+ \_ = \_$ only inside its domain,
	and any combination of two vectors from $X$ falls outside of $X$.
	However, the isomorphism $f$ does not extend to an automorphism of $\vatoms$,
	since vectors in $Y$ are not independent%
	\footnote{The notion of homogeneity can be extended to structures with relations and functions, but one must consider
	\emph{finitely-generated} induced substructures of $\atoms$ instead of finite ones.
	Note that $\vatoms$ becomes homogeneous if $+$ is considered as a binary \emph{function}, instead of a relation.
	The reason is that, in the presence of the functional symbol $+$,
	the homogeneity condition for $\vatoms$ quantifies over finite induced substructures that are closed w.r.t. $+$,
	unlike the substructures in our example.}.
	%
	%
\end{example}
It is worth mentioning that, while some atom structures are not homogenous,
sometimes adding extra relational symbols (thus restricting the notion of isomorphic substructure) can make it homogeneous;
cf. the example of universal tree order atoms from Sec.~\ref{sec:examples},
where adding one extra relational symbol turns a non-homogeneous structure it into a homogeneous one.



Fix a homogeneous relational structure $\atoms$.
We give a precise complexity upper-bound for the complexity of the concrete saturation procedure from Fig.~\ref{fig:concrete:saturation} and, thus, for reachability.
This depends on the complexity of the induced substructure problem for $\atoms$.
The \emph{(finite) induced substructure problem} for $\atoms$
asks whether a given finite structure $A$ over the same vocabulary is an induced substructure of $\atoms$.
This amounts to find an isomorphism mapping elements from $A$ into atoms $\atoms$ s.t. all relations from the vocabulary are preserved.
Assume that the induced substructure problem for $\atoms$ is decidable in time $T(k)$,
where $k$ is the size of the input.
The complexity estimations below are always understood with respect to the sizes of the representing formulas.
Let the \emph{width} of a formula be the number of its variables.
Let $n$ be the width of an input automaton, defined as the greatest width of the formulas appearing in its definition,
and let $m$ be its \emph{size}, defined as the sum of sizes of the defining formulas.
By \emph{$T$-relative pseudo-polynomial} time complexity we mean the time complexity
\begin{align*}
	2^{\poly(n)} \cdot \poly(m) \cdot T(\poly(n)),
\end{align*}
i.e., exponential in the width $n$ but polynomial in the size $m$.
Note that this is \emph{relative} to the complexity $T$ of the induced substructure problem.
%
%
\begin{theorem} \label{thm:pda decid homo}
	Let $\atoms$ be a homogeneous structure with induced substructure problem decidable in time $T(k)$.
	For \fodef \PDS $\pds$ over $\atoms$ and an \fodef \NFA $\nfa$ recognizing a regular set of configurations $L_{\pds}(\nfa)$,
	one can construct in $T$-relative pseudo-polynomial time
	an \fodef \NFA $\nfb$ recognizing $L_{\pds}(\nfb) = \invreach{\pds}{L_{\pds}(\nfa)}$.
\end{theorem}
%
%
As a consequence, reachability in \fodef \PDS over $\atoms$ is decidable in $T$-relative pseudo-polynomial time.
\begin{proof} 

	Fix a homogeneous relational structure $\atoms$,
	and suppose that its induced substructure problem is decidable in time $T(k)$.
	We show that the concrete saturation algorithm from Fig.~\ref{fig:concrete:saturation} terminates in $T$-relative pseudo-polynomial time.
	We use quantifier-free formulas over the vocabulary of $\atoms$ in \emph{legal disjunctive normal form}, to be defined below.
	A \emph{positive literal} is a predicate of the form $r(x_1, \ldots, x_k)$,
	where $x_1, \ldots, x_k$ are variables, and $r$ is a relational symbol in the vocabulary of $\atoms$.
	A \emph{negative literal} is the negation $\neg r(x_1, \ldots, x_k)$ of a positive literal,
	and a \emph{literal} is either a positive or a negative literal.
	We treat equality in the same way as other relations of $\atoms$,
	thus there are also equality and inequality literals.
	A \emph{clause} is a conjunction of pairwise different literals.
	A clause $\phi$ is \emph{complete} if, for every positive literal $l$ over the variables of $\phi$,
	either $l$ or its negation appears in $\phi$, but not both.
	A complete clause $\phi$ is \emph{consistent} if 
	\begin{itemize}
		\item the equality literals define an equivalence over the variables of $\phi$, and 
		\item the literals of $\phi$ are invariant under this equivalence relation,
		i.e., replacing variables appearing 
		in a literal of $\phi$ with equivalent ones yields a literal that also appears in $\phi$.
	\end{itemize} 
	A consistent clause $\phi$ gives rise to a finite structure $\Aa_\phi$ over the same vocabulary as $\atoms$, 
	whose elements are equivalence classes of variables,
	and where a relation $r([x_1], \dots, [x_k])$ holds if, and only if, $r(x_1, \dots, x_k)$ appears in $\phi$
	(the choice of representative variables is irrelevant since $\phi$ is consistent).
	Thus, valuations satisfying $\phi$ are in one-to-one correspondence with \emph{embeddings}
	of $\Aa_\phi$ into $\atoms$, by which we mean injective homomorphisms that both preserve and reflect relations.
	A consistent clause $\phi$ is \emph{legal} if, and only if, the structure $\Aa_\phi$ is isomorphic to
	an induced substructure of $\atoms$, i.e., if there exists an embedding of $\Aa_\phi$ into $\atoms$,
	written $\Aa_\phi \sqsubseteq \atoms$.
	Thus, a clause $\phi$ is legal if, and only if, it is satisfiable.
	\begin{proposition}
		\label{claim:legal}
		Legality of a complete clause of size $m$ is decidable in time $\poly(m) + T(m)$.
	\end{proposition} 
	We consider two clauses to be equal when they contain the same literals.
	A formula is in \emph{legal disjunctive normal form (ldnf)}
	if it is a disjunction of pairwise different legal clauses over the same variables.
	We use the convention that the empty clause and the empty ldnf represent, respectively, true and false.
	For two formulas $\phi$ and $\psi$ with the same free variables,
	we say that they are \emph{equivalent}, written $\phi \equiv \psi$, when $[\phi] = [\psi]$, i.e.,
	when they define the same set of tuples.
	%
	%
	\begin{proposition} \label{claim:qf-ldnf}
		A quantifier-free formula $\phi$ can be transformed
		into an equivalent formula $\psi$ in ldnf
		in $T$-relative pseudo-polynomial time.
	\end{proposition}
	\begin{proof}
		Enumerate exhaustively all complete clauses over the variables of $\phi$,
		and keep only those clauses $\{\psi_i\}_i$ which are legal (which is efficiently checkable by Proposition~\ref{claim:legal}),
		and that satisfy $\phi$ (computable in time polynomial in the size of $\phi$).
		Take $\psi = \bigvee_i \psi_i$.
		Clearly, $\psi \equiv \phi$.
		The time complexity claim follows since the number of complete clauses is exponential in the number of variables,
		but independent from the size of $\phi$.
	\end{proof}
	For homogeneous structures, the previous claim can be strengthened to first-order formulas.
	Essentially, this follows from the fact that, in a homogeneous structure,
	existential quantification can always be resolved positively.
	%
	\begin{proposition} \label{claim:fo-ldnf}
		A first-order formula $\phi$ can be transformed to an equivalent formula $\psi $ in ldnf
		in $T$-relative pseudo-polynomial time.
	\end{proposition}
	\begin{proof}
		As the first step, transform the input formula into prenex normal form. 
		Then, transform the quantifier-free subformula into an equivalent ldnf, using Proposition~\ref{claim:qf-ldnf}.
		Finally, eliminate the quantifiers in sequence, starting from the innermost one, 
		keeping the quantifier-free subformula in ldnf.
		Elimination of one existential quantifier is done as follows.
		First, distribute it over the disjunction of clauses,
		\[
		\phi \quad \equiv \quad \exists x \cdot \psi_1 \lor \ldots \lor \psi_n \quad \equiv \quad
		\exists x \cdot \psi_1 \ \lor \ \ldots \ \lor \ \exists x \cdot \psi_n
		\]
		and then replace every disjunct $\exists x \cdot \psi_i$ with the clause $\psi'_i$ obtained from $\psi_i$ by removing
		those literals that contain $x$.
		We claim that, after elimination of duplicates,
		\[
		\phi \quad \equiv \quad \psi'_1 \ \lor \ \ldots \ \lor \ \psi'_{n'} \ ,
		\]
		where the right-hand side is in ldnf.
		To this end, we show that each $\psi'_i$ is legal,
		and that $\exists x \cdot \psi_i \equiv \psi'_i$.
		Let $\Aa_{\psi_i}$ and $\Aa_{\psi'_i}$ be the two substructures of $\atoms$ defined by the two clauses.
		Clearly, $\Aa_{\psi'_i} \sqsubseteq \Aa_{\psi_i} \sqsubseteq \atoms$,
		which immediately implies legality of $\psi'_i$ by transitivity.
		The left-to-right inclusion $\defin{\exists x \cdot \psi_i} \subseteq \defin{\psi'_i}$
		of the equivalence between $\exists x \cdot \psi_i$ and $\psi_i'$ is immediate,
		since $\exists x \cdot \psi_i$ is more discriminating.
		For the other inclusion $\defin{\psi'_i} \subseteq \defin{\exists x \cdot \psi_i}$,
		let $\bar a' \in \defin{\psi'_i}$.
		Let $f_{\bar a'}$ be the natural embedding of $\Aa_{\psi'_i}$ into $\atoms$
		mapping each equivalence class of variables in $\Aa_{\psi'_i}$ to the corresponding element in $\bar a'$.
		Similarly, since $\Aa_{\psi_i} \sqsubseteq \atoms$,
		there exists a tuple $\bar a b$ and an embedding $g_{\bar a b}$ of $\Aa_{\psi_i}$ into $\atoms$,
		where $g_{\bar a b}([x]) = b$.
		The substructure induced by $\bar a$ is isomorphic to that induced by $\bar a'$.
		Let $h$ be such an isomorphism.
		Since $\atoms$ is homogeneous, $h$ extends to a full automorphism of $\atoms$.
		Define $b' = h(b)$.
		Then, $\bar a' b' \in \defin{\psi_i}$,
		and thus $\bar a' \in \defin{\exists x \cdot \psi_i}$.
		
		%
		
%
%
		
		The universal quantifier is handled with the equivalence
		$\forall x \cdot \phi \equiv \neg \exists x \cdot \neg \phi$:
		First we replace $\neg \phi$ by an equivalent formula in ldnf $\psi$ by applying Proposition~\ref{claim:qf-ldnf}.
		Then, we apply the procedure above to remove the existential quantifier in $\exists x \cdot \psi$,
		and we thus obtain another formula $\psi'$ in ldnf s.t. $\exists x \cdot \neg \phi \equiv \psi'$.
		Finally, a further application of Proposition~\ref{claim:qf-ldnf} to $\neg \psi'$ yields
		a formula $\psi''$ in ldnf s.t. $\psi'' \equiv \neg \exists x \cdot \neg \phi$.
	\end{proof}

	By repeatedly using Proposition~\ref{claim:fo-ldnf},
	we can implement the saturation algorithm in $T$-relative pseudo-polynomial time:
	First, transform all the formulas defining states and transitions of the input automata $\pds$ and $\nfa$ into ldnf.
	Then, in every iteration, the formula $\forced(\delta')$ is also transformed into ldnf.
	Step (2) is implemented by computing the union of clauses,
	and the implication in step (3) reduces to the inclusion of the sets of clauses of $\forced(\delta')$ into those of $\delta'$.
	Thus, one iteration of the algorithm requires relative pseudo-polynomial time.
	The total number of iterations is bounded by the number of orbits of the set $Q \times \Gamma \times Q$,
	since in every iteration at least one orbit is added to $\delta'$.
	By Proposition~\ref{claim:nr of orbits},
	the number of orbits in bounded by $2^{\poly(n)}$ 
	where $n$ is the dimension of $Q \times \Gamma \times Q$.
	Therefore, the concrete saturation algorithm runs in $T$-relative pseudo-polynomial time for homogeneous atoms.
	\ignore{
	\mysubsection{Proof of Proposition~\ref{prop:lower bound}}
	We sketch the proof of \exptime-hardness of the PDA reachability problem for equality atoms.
	For any other atoms, the reachability problem is at least as hard as for equality atoms.

	The proof is by reduction from the following
	decision problem: given $n$ (classical) NFA $\baut_1, \ldots, \baut_n$ and a (classical) PDA $\aaut$,
	decide whether the intersection of the languages of all the input automata is nonempty.
	From the automata $\baut_1, \ldots, \baut_n$ and $\aaut$ we build an \fodef PDA $\bar \aaut$
	over the same input alphabet, with the same stack alphabet as $\aaut$.
	The control states (locations) of $\bar \aaut$ are states of $\aaut$.
	The state space of $\bar \aaut$ has dimension $m+1$ (intuitively, $m+1$ registers), 
	where $m$ is the sum of sizes of state spaces of $\baut_1, \ldots, \baut_n$. 
	One register is used as a reference, and by equality of other register with the reference one we encode states of
	the automata $\baut_1, \ldots, \baut_n$. By a first-order formula of polynomial size we describe the 
	transition rules of $\bar \aaut$, which simulate the synchronous transitions of the input automata
	$\baut_1, \ldots, \baut_n$ and $\aaut$.
	The accepting states may be easily defined for $\bar \aaut$ so that
	the language of $\bar \aaut$ is equal to the intersection of the languages of $\baut_1, \ldots, \baut_n$ and $\aaut$,
	and therefore our construction preserves nonemptiness.
	The nonemptiness of $\bar \aaut$, in turn, may be easily expressed as a special case of the reachability problem.

	}
\end{proof}

As a consequence of Theorem~\ref{thm:pda decid homo},
under a bound on the width of input automata,
the \PDS reachability problem is in \ptime,
independently of the complexity $T(k)$ of the induced substructure problem.
Moreover, the proof of Theorem~\ref{thm:pda decid homo} reveals that
the polynomial above does not depend on the bound on width\footnote{We are
grateful to Mikołaj Bojańczyk for noticing this fact.}.
\begin{corollary}
	The \PDS reachability problem is fixed-parameter \ptime,
	with the width of the input automaton as the parameter.
\end{corollary}

In Theorem~\ref{thm:pda decid homo} we have shown that
the complexity of the saturation procedure/reachability can be upper-bounded
once we have a bound on the complexity of the induced substructure problem.
We show below that, depending on the homogeneous structure,
the latter problem (and thus reachability)
can be of arbitrarily high complexity, or even undecidable.
Therefore, the bound on the time complexity of induced substructure problem in Theorem~\ref{thm:pda decid homo} is a necessary assumption.
\begin{theorem}
	\label{thm:generic hardness}
	Let $X \subseteq \nat$ be a set of natural numbers.
	There exists a homogeneous structure $\atoms_X$
	s.t. membership in $X$ is many-one reducible to the induced substructure problem for $\atoms_X$.
\end{theorem}
\begin{proof}
	\newcommand{\TT}{{\cal T}}
	
	%
	Let $X \subseteq \nat$ be an arbitrary set of natural numbers.
	Intuitively, we effectively encode the set of natural numbers in an infinite antichain of finite tournaments,
	and we construct a homogeneous structure $\atoms_X$
	s.t., for every natural number $n \in \nat$,
	$n \in X$ if, and only if, the encoding of $n$ is an induced substructure of $\atoms_X$.	
	We use the instantiation of the embedding partial order $\sqsubseteq$ to finite directed graphs:
	$G \sqsubseteq H$ if $G$ is isomorphic to an induced subgraph of $H$.
	%
	A \emph{tournament} is a directed graph $T = (V, E)$ s.t., for every pair of vertices $x, y \in V$,
	either $(x, y) \in E$, or $(y, x) \in E$, but not both.
	It is known that there exists a countably infinite $\sqsubseteq$-antichain $\TT$ of finite tournaments \cite{Henson72}.
	Let $f$ be an efficiently computable bijective mapping between natural numbers and tournaments in the
	antichain $\TT$.
	Let $\TT_X$ be those finite tournaments $T$ in $\TT$ with $T = f(n)$ for some $n \in X$.
	The construction of $\atoms_X$ uses the following result.
	\begin{proposition}[\cite{survey}; see also~\cite{Henson72}]
		\
		For every $\sqsubseteq$-upward-closed family $\TT$ of finite tournaments,
		there is a homogeneous directed graph $\atoms$ such that,
		for every finite tournament $T$,
		%
			$T \sqsubseteq \atoms \textrm{ if, and only if, } T \in \TT.$
	\end{proposition}
	Let $\atoms_X$ be the homogeneous directed graph obtained by applying the proposition above
	to the upward closure of the antichain $\TT_X$.
	Then, for a natural number $n \in \nat$,
	we have $n \in X$ if, and only if, the finite tournament $f(n)$ is in $\TT_X$,
	which is the same as $f(n)$ being in the upward-closure of $\TT_X$,
	since $f(n)$ is by construction in the antichain $\TT$.
	By the proposition above, the latter property is equivalent to ask whether $f(n) \sqsubseteq\atoms_X$.
	Therefore, we can reduce membership in $X$ to the induced substructure problem in $\atoms_X$.
\end{proof}

%% file: examples.tex

\section{Examples of homogeneous structures}

\label{sec:examples}

The purpose of this section is to provide concrete examples of homogeneous structures
for which we can efficiently solve the reachability problem of \fodef \PDS.
Those are well known in the model-theoretic community (cf. \cite{survey}),
and we present them here in order to show the wide applicability of our results.
We also present a general technique, called \emph{wreath product}, 
which can be used to derive new homogeneous structures from known ones.
Recall that, by Theorem~\ref{thm:pda decid homo},
if $T(k)$ is the time complexity of the induced substructure problem of a homogeneous structure $\atoms$,
then reachability of \fodef \PDS over $\atoms$ is decidable in $T$-relative pseudo-polynomial time.
When the former problem is in \ptime, reachability can be solved in \exptime
by the following corollary of Theorem~\ref{thm:pda decid homo}.
\begin{corollary} \label{cor:pda exptime homo}
	Let $\atoms$ be a homogeneous relational structure with a \ptime induced substructure problem.
	For \fodef \PDS $\pds$ over $\atoms$ and an \fodef \NFA $\nfa$ recognizing a regular set of configurations $L_{\pds}(\nfa)$,
	one can construct in \exptime
	an \fodef \NFA $\nfb$ recognizing $L_{\pds}(\nfb) = \invreach{\pds}{L_{\pds}(\nfa)}$.
	In particular, the \fodef \PDS reachability problem over $\atoms$ is in \exptime.
\end{corollary}
All the concrete examples that we provide in the sequel, and all infinitely many examples that can be obtained by applying the wreath product,
have a \ptime induced substructure problem, and thus reachability is in \exptime.

\subparagraph{Equality.}

Equality atoms $(\D, =)$ consist of a countably-infinite set $\D$ together with the equality relation.
Automorphisms are permutations of $\D$.
Homogeneity follows from the fact that any finite partial bijection $\D \to \D$ 
can be extended to a permutation of the whole set $\D$.
This is arguably the simplest homogeneous structure.
The induced substructure problem is in \ptime,
since it amounts to check whether the interpretation of $=$ in a given finite structure is the equality relation.
By Corollary~\ref{cor:pda exptime homo}, reachability for \fodef \PDS over equality atoms is in \exptime.
This subsumes the result of~\cite{MRT14}, which considers a special case of our model where, among other restrictions, 
the input and stack alphabets are 1-dimensional, and the transition relation is quantifier-free definable (instead of \fodef).
%
Additionally, \cite{MRT14} shows that the problem is \exptime-hard for equality atoms.

All the examples below generalize equality atoms by adding more relations to the vocabulary.
We omit equality, which is assumed to always be in the vocabulary.

\subparagraph{Equivalence.}

\emph{Equivalence atoms} $(\D, R)$ consist of a countably-infinite set $\D$
and an infinite-index equivalence relation $R$ over $\D$
s.t. each one of the infinitely-many equivalence classes is itself an infinite subset of $\D$.
An automorphism of equivalence atoms is a bijection $f$ of $\D$ which respects $R$,
in the sense that, for every $x, y \in \D$, $(x, y) \in R$ if, and only if, $(f(x), f(y)) \in R$.
Equivalence atoms are homogeneous.
(We will see later that equivalence atoms are isomorphic with the wreath product of equality atoms with itself.)
%
This can model hierarchically nested data,
where one can check whether two elements belong to the same equivalence class,
and, if so, whether they actually are the same element.
Higher nested equivalence atoms can be obtained by iterating this process:
$0$-nested equivalence atoms are just equality atoms; and 
for any $k \geq 0$, $(k+1)$-nested equivalence atoms can be seen as the disjoint union
of infinitely many copies of $k$-nested equivalence atoms, 
with one additional equivalence relation that relates a pair of elements iff they belong to the same copy.

\subparagraph{Total, betweenness, and cyclic order.}

\emph{Total order atoms} $(\Q, \leq)$ can be presented as the rational numbers $\Q$ together with the natural total order $\leq$.
Automorphisms are monotonic bijections of rational numbers.
Homogeneity follows from the fact that $\leq$ is dense:
A monotonic bijection $f : X \to Y$ over a finite domain $X$ extends to an automorphism of $\Q$.
The induced substructure problem is in \ptime,
since it amounts to check whether the interpretation of $\leq$ in a given finite structure is a total order.
This can be used to model qualitative time,
where events are totally ordered, but no information is available on the distance between them.
Another instance is given by data-centric applications \cite{DeutschHullPatriziVianu:ICDT:2009}.

\emph{Betweenness order atoms} $(\Q, B)$ use the betweenness relation $B$,
which is obtained by considering the order $\leq$ up to reversal:
$B(x, y, z)$ holds when $x$ lies between $y$ and $z$, i.e., either $y < x < z$ or $z < x < y$.
This can be used to model time where one is not interested on the order between the events themselves,
but rather on whether an event happened between two other events.
\emph{Cyclic order atoms} $(\Q, K)$ use the ternary cyclic ordering $K$
obtained by bending the total order into a circle.
Formally, $K(x, y, z)$ if either $x < y < z$, or $z < x < y$, or $y < z < x$.
This can model a notion of qualitative cyclic time, where events cyclically repeat,
but no precise timing information is available.
For both betweenness and cyclic order atoms, the induced substructure problem is in \ptime.

\subparagraph{Universal partial order and preorder.}

Every relational homogeneous structure is obtained as the \emph{Fraiss\'e limit} of the set of all its finite induced
substructures~\cite{F53}.
(We do not formally define here the notion of Fraiss\'e limit, which is a central tool for constructing homogeneous structures; cf. \cite{survey}.)
For instance, total order atoms are the Fraiss\'e limit of all finite total orders.
\emph{Partial order atoms} are obtained as the Fraiss\'e limit of the set of all finite partial orders.
The induced substructure problem amounts to determine whether the interpretation of $\leq$ in a given finite structure is a partial order,
which can clearly be done in \ptime.
This can be used to model the ordering of events in distributed systems.
Along the same lines one obtains \emph{preorder atoms}.

\subparagraph{Universal tree order.}

A \emph{tree order} (or semilinear order) is a partially ordered structure $(A, \leq)$
s.t. a) every two elements have an common upper bound,
and b) for every element, its upward closure is totally ordered.
Tree order atoms $(T, \leq)$ are obtained as the Fraiss\'e limit of the set of all finite tree orders.
Intuitively, tree order atoms consists of a countably-infinite tree order where each maximal path is isomorphic to total order atoms.
Tree order atoms as presented here are not homogeneous.
Intuitively, this happens because isomorphic substructures have least upper bounds outside the structures themselves,
and they might relate to those in an incomparable way.
%
This can be amended by introducing be the following ternary relation: 
$R(x, y, z)$ holds when the lub of $x$ and $y$ is incomparable with $z$.
Then, $(T, \leq, R)$ is homogeneous,
and it can be obtained as the Fraiss\'e limit of the set of all extended finite tree orders
$(A, \leq, R)$.
The induced substructure problem is in \ptime for $(T, \leq, R)$.

\subparagraph{Universal graph and tournament.}

\emph{Universal graph atoms} are obtained as the Fraiss\'e limit of the set of all finite graphs.
This is also known as \emph{Rado's graph} or \emph{the random graph}.
The induced substructure problem is trivial since the universal graph contains an isomorphic copy of every finite graph.
Similarly, \emph{universal tournament atoms} are the Fraiss\'e limit of the set of all finite tournaments,
where a tournament is an irreflexive graph $T = (V, E)$ s.t.,
for every two nodes $x, y \in V$, either $(x, y) \in E$, or $(y, x) \in E$.
Given a graph, it is clearly checkable in \ptime whether it is actually a tournament,
thus the induced substructure problem is in \ptime also in this case.


\subparagraph{Wreath products.}

We conclude this section by giving a construction which allows to compose homogeneous structures in order to produce new ones.
Given two relational structures $\A = (A, R_1, \dots, R_m)$ and $\B = (B, S_1, \dots, S_n)$,
their \emph{wreath product} is the relational structure $\wreath \A \B = (A \times B, R'_1, \dots, R'_m, S'_1, \dots, S'_n)$,
where $((a_1, b_1), \dots, (a_k, b_k)) \in R'_i$ if $(a_1, \dots, a_k) \in R_i$,
and $((a_1, b_1), \dots, (a_k, b_k)) \in S'_j$ if $a_1 = \cdots = a_k$ and $(b_1, \dots, b_k) \in S_j$.
Intuitively, $\wreath \A \B$ is obtained by replacing each element in $\A$ with a disjoint copy of $\B$.
It can be checked that, if the two structures $\A$ and $\B$ are homogeneous,
then the same holds for their wreath product $\wreath \A \B$.   
The induced substructure problem for $\wreath \A \B$ reduces in \ptime to the same problem for $\A$ and $\B$:
$\set{(a_1, b_1), \dots, (a_k, b_k)}$ is an induced substructure of $\wreath \A \B$ if, and only if,
$\set{a_1, \dots, a_k}$ is an induced substructure of $\A$,
and for every $i$, $\setof{b_j}{a_j = a_i}$  is an induced substructure of $\B$.
%
Therefore, if both $\A$ and $\B$ have a \ptime induced substructure problem,
then the same holds for $\wreath \A \B$, and Corollary~\ref{cor:pda exptime homo} applies.

As an application of the wreath product, take $\A_0 = (\D, =)$ to be equality atoms, and,
for each $k \geq 0$, let $\A_{k+1} = \wreath {\A_0} {\A_k}$.
Then, $\A_1$ is just the equivalence atoms presented before,
and, more generally, $\A_k = (\D, R_1, \dots, R_k)$ is \emph{$k$-nested equivalence atoms},
which can be used to model data with nested equivalence relations.
%
For each of those infinitely many examples, the reachability problem for \fodef \PDS is in \exptime.
%

%% file: conclusions.tex

\section{Conclusions}

\label{sec:conclusions}

We have studied the reachability problem for a model of PDS
with countably-infinite \fodef states, stack alphabet, and transitions relation.
We advocate a Ockham's razor research strategy that refrains from inventing seemingly new notions.
Instead, we have taken the standard definition of PDS
and re-interpreted it in the richer framework of \fodef sets instead of ordinary finite sets.
This covers the well-known model of pushdown register automata \cite{ChengKaminski:CFL:AI98,MRT14}
as one instantiation of the general paradigm,
and we have shown that the optimal \exptime complexity for the reachability problem for this model can be recovered in the more general framework.
This same paradigm can of course be applied to a variety of different models,
like timed \PDS~\cite{AbdullaAtigStenman:DensePDA:12},
data/timed extensions of Petri nets~\cite{AbdullaN01,LazicNewcombOuaknineRoscoeWorrell:2007}, 
lossy channel systems~\cite{AbdullaAtigCederberg:TLCS:FSTTCS12},
1-clock/1-register alternating automata~\cite{LasotaWalukiewicz:ATA:ACM08,OuaknineWorrel:Inclusion:LICS04,DemriLazic:FreezeLTL:ACM09},
rewriting systems~\cite{BouajjaniHabermehlJurskiSighireanu:Rewriting:FCT:2007}, etc.
Therefore, the present paper can be seen as a proof of concept of the new research strategy.
For example, one could consider \emph{\fodef pushdown automata} (PDA) and \emph{\fodef context-free grammars} (CFG)
as acceptors of languages over infinite alphabets.
The definition of \fodef PDA is analogous to PDS,
except that the transition relation is an \fodef subset of
$Q \times \Gamma^* \times A_\eps \times Q \times \Gamma^*$,
where $A_\eps = A \cup \{\eps\}$ is an \fodef alphabet extended with the empty word.
Similarly, \fodef CFG can be defined as \emph{stateless} \fodef PDA
where every transition pops exactly one symbol from the stack.
It is easy to prove that \fodef PDA languages coincide with \fodef context-free languages for oligomorphic atoms \cite{BKL11full},
and that the latter are closed under union, concatenation, Kleene star, homomorphism, inverse homomorphism, intersection with \fodef regular languages,
and that collapsing each orbit to a different symbol yields a classical context-free language. 